\documentclass[11pt]{article}

\usepackage{geometry}
\usepackage{cite}
\usepackage{paralist}
\usepackage{graphicx}
\usepackage{subfigure}
\usepackage{amssymb}
\usepackage{amsmath}
\allowdisplaybreaks[1]

\usepackage{amsthm}

\newtheorem{theorem}{Theorem}
\newtheorem{lemma}[theorem]{Lemma}

\makeatletter
\def\@endtheorem{\endtrivlist}
\makeatother

\newcounter{rules}
\newcommand{\Rule}{\refstepcounter{rules}\par\smallskip\noindent
\textbf{(\arabic{rules})}
\quad}
\newcounter{rulesb}
\newcommand{\Ruleb}{\refstepcounter{rulesb}\par\smallskip\noindent
\textbf{(B\arabic{rulesb})}\quad}

\newcommand{\wvcalgname}{WVCAlg}
\newcommand{\wvcalg}[3]{\mathrm{WVCAlg}(#1,#2,#3,f)}

\newcommand{\cc}[3]{\mathcal{C}_{#1}(#2,#3)}
\newcommand{\ccs}[3]{\mathcal{C}_{#1}^*(#2,#3)}
\newcommand{\vcc}[3]{V_{#1}(#2,#3)}
\newcommand{\vccs}[3]{V_{#1}^*(#2,#3)}

\newcommand{\map}[2]{h_{#1,#2}}
\newcommand{\state}{q}

\newcommand{\ca}{1.402}
\newcommand{\cb}{0.9808}
\newcommand{\calpha}{0.156}
\newcommand{\cbeta}{0.175}

\begin{document}

\title{Weighted vertex cover on graphs with maximum degree 3}
\author{Dekel Tsur%
\thanks{Department of Computer Science, Ben-Gurion University of the Negev.
Email: \texttt{dekelts@cs.bgu.ac.il}}}
\date{}
\maketitle

\begin{abstract}
We give a parameterized algorithm for weighted vertex cover on graphs with
maximum degree 3 whose time complexity is $O^*(\ca^t)$,
where $t$ is the minimum size of a vertex cover of the input graph.
\end{abstract}


\section{Introduction}
For an undirected graph $G$, a \emph{vertex cover} of $G$ is a set of vertices
$S$ such that every edge of $G$ is incident on at least one vertex in $S$.
In the \emph{vertex cover problem} the input is an undirected graph $G$
and the goal is to find a vertex cover of $G$ with minimum size.
In the \emph{weighted vertex cover problem} the input is an undirected graph $G$
and a weight function $w\colon V(G) \to \mathbb{R}^{\geq 0}$.
The goal is to find a vertex cover of $G$ with minimum weight.

The parameterized complexity of the vertex cover problem have been studied
extensively.
For the unweighted problem, the first parameterized algorithm
was given in~\cite{buss1993nondeterminism}.
Improved algorithms were given in~\cite{downey1995parameterized, %
balasubramanian1998improved,downey1999parameterized,niedermeier1999upper, %
stege1999improved,chen2001vertex,chen2010improved, %
niedermeier2003efficient,chandran2004refined}.
The unweighted problem was also studied on graphs with maximum degree
$3$~\cite{chen2000improvement,chen2005labeled,razgon2009faster,xiao2010note}.
For the weighted vertex cover problem, Niedermeier et
al.~\cite{niedermeier2003efficient} gave an algorithm with time complexity
$O^*(1.396^W)$, where $W$ is the minimum weight of a vertex cover of the input
graph.
They also gave an algorithm with exponential space whose running time is
$O^*(1.379^W)$.
Fomin et al.~\cite{fomin2006branching} gave an algorithm with exponential
space whose running time is $O^*(1.357^W)$.
Shachnai and Zehavi~\cite{shachnai2017multivariate} 
gave a polynomial space algorithm with time complexity $O^*(1.381^s)$,
where $s\leq W$ is the minimum size of a minimum weight vertex cover of the
input graph,
and an exponential space algorithm with time complexity $O^*(1.363^s)$.
Additionally, they gave an algorithm with time complexity $O^*(1.443^t)$,
where $t \leq s$ is the minimum size of a vertex cover of the input graph,
and an algorithm for graphs with maximum degree 3 whose time complexity
is $O^*(1.415^t)$. 

In this paper, we give an algorithm for weighted vertex cover on graph
with maximum degree 3 whose time complexity is $O^*(\ca^t)$.

\section{Preliminaries}\label{sec:preliminaries}
For a graph $G$, let $V(G)$ and $E(G)$ denote the sets of vertices and edges of
$G$, respectively.
For a graph $G$ and a vertex $v\in V(G)$,
$N(v) = \{u \in V(G)\colon (u,v)\in E(G)\}$ and
$N[v] = N(v) \cup \{v\}$.
For a set of vertices $S$,
$N(S) = (\bigcup_{v\in S} N(v)) \setminus S$.

For a graph $G$ and a set of vertices $S$, $G[S]$ is the subgraph
of $G$ induced by $S$ (namely, $G[S]=(S,E(G)\cap (S\times S))$).
We also define $G-S = G[V\setminus S]$.
For a vertex $v$, we write $G-v$ instead of $G-\{v\}$.

For a graph $G$ and a set $U \subseteq V(G)$, let
$\cc{i}{G}{U}$ be the set of connected components $C$ of $G[U]$ such that
$|C|=i$ and $G[C]$ is a complete graph.
Let $\ccs{i}{G}{U}$ be the set of all $C \in \cc{i}{G}{U}$ such that
there is no $v \in V(G)\setminus U$ that is adjacent to all the vertices of $C$.
Let $\vcc{i}{G}{U} = \bigcup_{C \in \cc{i}{G}{U}} C$,
$\vccs{i}{G}{U} = \bigcup_{C \in \ccs{i}{G}{U}} C$, and
$\vcc{\geq 2}{G}{U} = U \setminus \vcc{1}{G}{U}$.

\section{The algorithm}

Our algorithm is based on the algorithm of Shachnai
and Zehavi~\cite{shachnai2017multivariate}.
We first describe the algorithm of Shachnai and Zehavi
(we note that we describe the algorithm slightly differently
than~\cite{shachnai2017multivariate}).
Given an instance $G^*,w^*$ of weighted vertex cover,
the algorithm first finds a vertex cover $U^*$ of $G^*$ with minimum size
(using a fixed parameter algorithm for the unweighted problem).
Additionally, $U^*$ has the property that there is a mapping
$f\colon \ccs{3}{G^*}{U^*} \to \cc{2}{G^*}{U^*}$
such that for every $C \in \ccs{3}{G^*}{U^*}$ there is a vertex
$v \in V(G^*) \setminus U^*$ that is adjacent to a vertex in $C$ and to the two
vertices of $f(C)$.
The algorithm then calls $\wvcalg{G^*}{U^*}{w^*}$,
where $\wvcalg{G}{U}{w}$ is a recursive procedure that returns a minimum weight
vertex cover of $G$. The parameter $U$ is a vertex cover of $G$
(not necessarily a vertex cover with minimum size).
In particular, $V(G) \setminus U$ is an independent set of $G$.

Before describing procedure \wvcalgname,
we define two \emph{base branching rules}:

\Ruleb
Let $v \in U$.
Return a set of minimum weight among
$\wvcalg{G-v}{U \setminus \{v\}}{w} \cup \{v\}$
and
$\wvcalg{G-N[v]}{U \setminus N[v]}{w} \cup N(v)$.

\Ruleb
Let $C \in \cc{3}{G}{U}$.
Return a set of minimum weight among
$\wvcalg{G-A}{U \setminus A}{w} \cup A$,
where the minimum is taken over every set $A \subset C$ of size 2.

Procedure \wvcalgname\ is composed of reduction and branching rules.
The procedure applies the first applicable rule from the following rules
(the last rule is numbered~\ref{rule:final} in order to leave space for the rules
of our algorithm).

\Rule
If $G$ is bipartite, compute a minimum weight vertex cover $S$ of $G$
and return $S$.

\Rule
\label{rule:small-cc}
If there is a connected component $S$ in $G$ of size at most $10$,
compute a minimum weight vertex cover $S$ of $G[C]$
and return $\wvcalg{G-C}{U\setminus C}{w}\cup S$.

\Rule
\label{rule:no-U-neighbors}
If there is $v\in U$ such that $v$ has no neighbors in $V(G) \setminus U$,
return $\wvcalg{G}{U \setminus \{v\}}{w}$.

\Rule
\label{rule:largest}
If there is $C \in \cc{3}{G}{U}$ such that $f(C')\neq f(C)$
for every $C' \in \cc{3}{G}{U} \setminus \{C\}$,
choose $v\in f(C)$.
Apply Rule~(B1) on $v$.
In the branch $G-v$ apply Rule~(B2) on $C$.

\Rule
If $\cc{3}{G}{U} \neq \emptyset$, choose distinct $C,C' \in \cc{3}{G}{U}$ such
that $f(C)=f(C')$, and choose $v\in f(C)$.
Apply Rule~(B1) on $v$.
In the branch $G-v$ apply Rule~(B2) on $C$, and in each resulting branch,
apply Rule~(B2) on $C'$.

\Rule
\label{rule:path}
If there is $u \in U$ such that $|N(u) \cap U| = 1$ and $|N(v)\cap U| = 2$,
where $v$ is the unique neighbor of $u$ in $G[U]$,
apply Rule~(B1) on $v$.

\Rule
\label{rule:cycle}
If there is $v\in U$ such that  $|N(v) \cap U| = 2$,
apply Rule~(B1) on $v$.

\addtocounter{rules}{5}
\Rule
\label{rule:final}
Choose $v \in U$ such that $|N(v) \cap U| = 1$,
and apply Rule~(B1) on $v$.
\addtocounter{rules}{-6}

The analysis of the algorithm above uses the \emph{measure and conquer}
technique, using the measure function $m(G,U) = |\vcc{\geq 2}{G}{U}|$.
The analysis shows that the number of leaves in the branching tree of $G^*$
is at most $1.415^{m(G^*,U^*)} \leq 1.415^t$, where $t = |U^*|$ is the
minimum size of a vertex cover of $G^*$.

\renewcommand{\wvcalg}[3]{\mathrm{WVCAlg}(#1,#2,#3,f,\state)}
\newcommand{\wvcalgb}[4]{\mathrm{WVCAlg}(#1,#2,#3,f,#4)}

We now describe our algorithm.
A \emph{good vertex cover} of a graph $G$ is a vertex cover $U$ of $G$ such that
every connected component in $G[U]$ has size at most~2.
Suppose that $U$ is a good vertex cover of $G$.
We say that a vertex $x \in V(G)\setminus U$ is \emph{bad}
if $|N(x) \cap \vcc{1}{G}{U}| \geq 1$ and either
$|N(x) \cap \vcc{2}{G}{U}| = 1$ or
$x$ is adjacent to both vertices of a connected component in $\cc{2}{G}{U}$.
We say that $x$ is \emph{semi-bad} if $x$ is not bad,
$|N(x) \cap \vcc{1}{G}{U}| = 1$ and $|N(x) \cap \vcc{2}{G}{U}| = 2$.
If $x$ is not bad or semi-bad, way say that $x$ is \emph{good}.
We define a mapping
$\map{G}{U}\colon \vcc{2}{G}{U} \to
 \{0, \frac{1}{4}, \frac{1}{2}, \frac{3}{4}, 1\}$ as follows.
If $v \in \vccs{2}{G}{U}$,
$\map{G}{U}(v) = \max(0,1-b_1-b_2/2)$, where
$b_1$ (resp., $b_2$) is the number of bad (resp., semi-bad) neighbors of $v$.
Now consider a vertex $v \in \vcc{2}{G}{U} \setminus \vccs{2}{G}{U}$,
and let $v'$ be the unique neighbor of $v$ in $G[U]$.
Then, $\map{G}{U}(v) = \map{G}{U}(v') = \max(0,1-b_1/2-b_2/4)$, where
$b_1$ (resp., $b_2$) is the number of bad (resp., semi-bad) vertices in
$N(\{v,v'\})$.

Our algorithm is based on the algorithm of Shachnai and Zehavi.
We make two changes to procedure \wvcalgname.
First, the procedure receives an additional parameter
$\state \in \{0,1,2\}$, which is initially $0$.
Additionally, the following rules are added.

\Rule
\label{rule:degree-1}
If there is $v \in V(G)$ with degree~1,
let $u$ be the neighbor of $v$.
If $w(v) \geq w(u)$, return
$\wvcalg{G-\{u,v\}}{U \setminus \{u,v\}}{w} \cup \{u\}$.
Otherwise, let $w'\colon V(G-v) \to \mathbb{R}^{\geq 0}$ be a function
in which $w'(u) = w(u) - w(v)$ and $w'(x) = w(x)$ for every $x \neq u$.
Let $S = \wvcalg{G-v}{U \setminus \{v\}}{w'}$.
If $u\in S$ return $S$ and otherwise return $S\cup \{v\}$.

\Rule
\label{rule:triangle}
If there is a triangle $v_1,v_2,v_3$ such that
either $v_1$ and $v_2$ have degree~2, or
there is a vertex $v_4$ such that $N(v_4) = \{v_1,v_2\}$,
return $\wvcalg{G-v_i}{U \setminus \{v_i\}}{w} \cup \{v_i\}$,
where $v_i$ is the vertex with minimum weight among $v_1$ and $v_2$.

\Rule
\label{rule:set-state}
If $\state = 0$, then $\state' \gets 1$ if
$|\vcc{1}{G}{U}| \geq \beta \cdot |\vcc{2}{G}{U}|$
and $\state' \gets 2$ otherwise, where $\beta = \cbeta$.
Return $\wvcalgb{G}{U}{w}{\state'}$.

\Rule
\label{rule:vu1u2}
If $\state = 2$ and there is $v \in \vccs{2}{G}{U}$ with $\map{G}{U}(v) > 0$
then let $v'$ be the unique neighbor of $v$ in $G[U]$.
If $|N(v)\setminus \{v'\}| = 1$ denote $N(v) = \{v',x_1\}$, and otherwise denote
$N(v) = \{v',x_1,x_2\}$ where $x_2$ is a good vertex.
Choose a vertex $u_1 \in (N(x_1) \setminus \{v\}) \cap \vcc{2}{G}{U}$,
and let $u'_1$ be the unique neighbor of $u_1$ in $G[U]$.
If $N(v) = \{v',x_1,x_2\}$, $u_1 \notin N(x_2)$, and
$N(x_2) \setminus \{v\} \neq \{u'_1\}$,
choose a vertex $u_2 \in N(x_2) \setminus\{v,u'_1\}$.
Now, apply Rule~(B1) on $u_1$.
If $u_2$ is defined, in each of the two branches obtained by the application
of Rule~(B1), apply Rule~(B1) on $u_2$.
This gives four branches whose graphs are
$G-\{u_1, u_2\}$, $G-(N[u_1] \cup \{u_2\})$, $G-(\{u_1\} \cup N[u_2])$, and
$G-(N[u_1] \cup N[u_2])$.

We now show the correctness of the rule.
We have that $|N(v)\setminus \{v'\}|\in \{1,2\}$
(this follows from the assumption that $G$ has maximum degree 3,
and the assumption that Rule~(\ref{rule:no-U-neighbors}) cannot be applied).
Therefore, either $N(v) = \{v',x_1\}$ or $N(v) = \{v',x_1,x_2\}$.
In the latter case, at least one of the vertices of $N(v)\setminus\{v'\}$ is
good since $\map{G}{U}(v) > 0$.
In both cases, the vertex $x_1$ is not bad (since $\map{G}{U}(v) > 0$).
The vertex $u_1$ exists since the assumption that Rule~(\ref{rule:degree-1})
cannot be applied implies that $N(x_1) \setminus \{v\} \neq \emptyset$,
and at least one of the vertices in $N(x_1) \setminus \{v\}$ is in
$\vcc{2}{G}{U}$ (since $x_1$ is not bad).
Note that $v \in \vccs{2}{G}{U}$ implies that $u_1 \neq v'$.

If $N(v) = \{v',x_1,x_2\}$, $u_1 \notin N(x_2)$, and
$N(x_2) \setminus \{v\} \neq \{u'_1\}$,
then the vertex $u_2$ exists:
If $u'_1 \notin N(x_2)$ then $u_2$ exists due to the assumption that
Rule~(\ref{rule:degree-1}) cannot be applied, and if $u'_1 \in N(x_2)$
then $u_2$ exist due to the assumption that
$N(x_2) \setminus \{v\} \neq \{u'_1\}$.
Since $x_2$ is good, $u_2 \in \vcc{2}{G}{U}$.
By definition, $u_2 \notin \{v,u_1,u'_1\}$.
Additionally, since $v \in \vccs{2}{G}{U}$, $u_2 \neq v'$.

Since Rule~(B1) is correct, and since the vertices $u_1,u'_1,u_2,u'_2$ are
distinct, it follows that Rule~(\ref{rule:vu1u2}) is correct.

\Rule
\label{rule:vu1u2-2}
If $\state = 2$ and there is $v \in \vcc{2}{G}{U} \setminus \vccs{2}{G}{U}$
with $\map{G}{U}(v) > 0$, then let $v'$ be the unique neighbor of $v$ in $G[U]$.
Let $x \in V(G) \setminus U$ be a vertex that is adjacent to both $v$ and $v'$.
Let $A = \{x \in N(\{v,v'\}) \colon N(x) \setminus \{v,v'\} \neq \emptyset\}$.
If $|A| = 2$, choose a non-bad vertex $x_1 \in A$.
Otherwise, choose $x_1,x_2 \in A$ such that $x_1$ is not bad and $x_2$ is good.
Choose $u_1 \in (N(x_1) \setminus \{v,v'\}) \cap \vcc{2}{G}{U}$,
and let $u'_1$ be the unique neighbor of $u_1$ in $G[U]$.
If $|A|=3$, $u_1 \notin N(x_2)$, and $N(x_2) \setminus \{v,v'\} \neq \{u'_1\}$,
choose a vertex $u_2 \in N(x_2) \setminus \{v,v',u'_1\}$
Now, apply Rule~(B1) on $u_1$.
If $u_2$ is defined, in each of the two branches obtained by the application
of Rule~(B1), apply Rule~(B1) on $u_2$.

We now show the correctness of the rule.
The existence of $x$ follows from the fact that $v \in \vccs{2}{G}{U}$.
Due to the assumption that Rule~(\ref{rule:degree-1}) and
Rule~(\ref{rule:triangle}) cannot be applied, $|A| \in \{2,3\}$, where
in the case $|A| = 3$ we have that $x\in A$.
If $|A| = 2$, there is at most one bad vertex in $A$
(since $\map{G}{U}(v) > 0$),
and therefore $A$ contains at least one non-bad vertex.
Otherwise ($|A|=3$), from the assumption that
$v\in \vcc{2}{G}{U} \setminus \vccs{2}{G}{U}$ we have that
there is a vertex $x \in V(G)\setminus U$ that is adjacent to both $v$ and $v'$.
Since $|A|=3$, $x\in A$.
By definition, $x$ is either bad or good, and therefore 
$A$ cannot contain three semi-bad vertices.
It follows that $A$ contains a good vertex, and an additional vertex that is
either good or semi-bad.

The vertex $u_1$ exists since $N(x_1) \setminus \{v,v'\} \neq \emptyset$ by
the definition of $A$,
and at least one of the vertices in $N(x_1) \setminus \{v,v'\}$ is in
$\vcc{2}{G}{U}$ (since $x_1$ is not bad).
If $|A|=3$, $u_1 \notin N(x_2)$, and $N(x_2) \setminus \{v,v'\} \neq \{u'_1\}$,
the vertex $u_2$ exists
(if $u'_1 \notin N(x_2)$ then $u_2$ exists due to the definition of $A$,
and if $u'_1 \in N(x_2)$ then $u_2$ exist due to the assumption that
$N(x_2) \setminus \{v,v'\} \neq \{u'_1\}$).
Since $x_2$ is good, $u_2 \in \vcc{2}{G}{U}$.
By definition, $u_2 \notin \{v,v',u_1,u'_1\}$.

Since Rule~(B1) is correct, and since the vertices $u_1,u'_1,u_2,u'_2$ are
distinct, it follows that Rule~(\ref{rule:vu1u2-2}) is correct.

Note that when the algorithm applies Rule~(\ref{rule:set-state})
in some branch, $U$ is a good vertex cover of $G$.
Therefore, the only rules that are applied afterward in the branch
are Rules (1), (2), (3), (\ref{rule:degree-1}), (\ref{rule:triangle}),
(\ref{rule:vu1u2}), (\ref{rule:vu1u2-2}), and (\ref{rule:final}).

We now analyze the time complexity of the algorithm.
We first analyze rules (1) to (\ref{rule:triangle}).
We use the measure function $m_1(G,U) = |\vcc{\geq 2}{G}{U}|+\alpha |\vcc{1}{G}{U}|$,
where $\alpha = \calpha$.
The branching rule with the largest branching number is
Rule~(\ref{rule:largest}).
This rule generates four branches. In each of the three branches that are
obtained from the branch $G-v$, $3$ vertices of $\vcc{\geq 2}{G}{U}$ are deleted
from the graph,
and $2$ vertices are moved from $\vcc{\geq 2}{G}{U}$ to $\vcc{1}{G}{U}$.
Therefore, the value of $m_1(G,U)$ decreases by $5-2\alpha$ in these branches.
In the branch $G-N[v]$ the algorithm applies Rule~(\ref{rule:no-U-neighbors})
on a vertex of $C$.
Therefore, in this branch one vertex of $\vcc{\geq 2}{G}{U}$ is deleted from the
graph,
one vertex is moved out of $\vcc{\geq 2}{G}{U}$,
and one vertex is moved from $\vcc{\geq 2}{G}{U}$ to $\vcc{1}{G}{U}$.
Therefore, the value of $m_1(G,U)$ decreases by $3-\alpha$.
Thus, the branching vector of Rule~(\ref{rule:largest}) is
$(5-2\alpha,3-\alpha)$, and the branching number is $\ca$.

The analysis of the other branching rules is similar to the analysis of these rules
in~\cite{shachnai2017multivariate}.
One exception is Rule~(\ref{rule:cycle}), for which we need a more careful
analysis.
Due to the previous rules, when this rule is applied, the connected component
of $v$ in $G[U]$ is a cordless cycle with at least 4 vertices, and denote the
vertices of this cycle by $v,v_1,v_2,\ldots,v_{|C|-1}$.
Therefore, after applying Rule~(\ref{rule:cycle}),
Rule~(\ref{rule:path}) can be applied in the branch $G-v$ on $v_2$.
Therefore, there are three branches: $G-\{v,v_2\}$, $G-(\{v\} \cup N[v_2])$,
and $G-N[v]$.
If the size of the cycle is $4$, in the first and third branches, 2 vertices
of the $\vcc{\geq 2}{G}{U}$ are deleted from the graph, and 2 vertices
of $\vcc{\geq 2}{G}{U}$ are moved to $\vcc{1}{G}{U}$.
In the second branch, 3 vertices of $\vcc{\geq 2}{G}{U}$ are deleted from the
graph,
and one vertex of $\vcc{\geq 2}{G}{U}$ is moved to $\vcc{1}{G}{U}$.
It follows that the branching vector is $(4-2\alpha,4-\alpha,4-2\alpha)$.
Similarly, the branching vector is $(3-\alpha,5-2\alpha,3-\alpha)$ if the size
of the cycle is 5.
If the size of the cycle is at least 6, Rule~(\ref{rule:path}) can
be applied in the branch $G-\{v,v_2\}$ on $v_4$, and
in the branch $G-N[v]$ on $v_3$.
Therefore, there are five branches:
$G-\{v,v_2,v_4\}$, $G-(\{v,v_2\} \cup N[v_4])$, $G-(\{v\} \cup N[v_2])$,
$G-(N[v] \cup \{v_3\})$, and $G-(N[v] \cup N[v_3])$.
The branching vector is at least as good as
$(5-2\alpha,6-2\alpha,4-\alpha,5-2\alpha,6-2\alpha)$.
The branching numbers of the branching vectors above are
1.342, 1.389, and 1.395, respectively.

We now consider a recursive call $\wvcalg{G'}{U'}{w'}$ in which
Rule~(\ref{rule:set-state}) is applied and
$|\vcc{1}{G'}{U'}| \geq \beta \cdot |\vcc{2}{G'}{U'}|$.
For the analysis of this branch, we switch from the measure function
$m_1(G,U)$ to a measure function $m_2(G,U) = (1+\alpha\beta)\cdot |\vcc{2}{G}{U}|$.
Note that $m_2(G',U') \leq m_1(G',U')$, so the switch is correct.
After Rule~(\ref{rule:set-state}) is applied, only Rules (1), (2), (3),
(\ref{rule:degree-1}), (\ref{rule:triangle}), and (\ref{rule:final})
are applied, and the only branching rule among these rules is
Rule~(\ref{rule:final}).
When Rule~(\ref{rule:final}) is applied, the value of $m_2(G,U)$
decreases by $2(1+\alpha\beta)$ in each branch (since in each branch, one vertex
in $\vcc{2}{G}{U}$ is deleted from the graph,
and one vertex in $\vcc{2}{G}{U}$ is moved to $\vcc{1}{G}{U}$).
The branching vector is $(2(1+\alpha\beta),2(1+\alpha\beta))$ and the branching
number is $\ca$.

We now consider a recursive call $\wvcalg{G'}{U'}{w'}$ in which
Rule~(\ref{rule:set-state}) is applied and
$|\vcc{1}{G'}{U'}| < \beta \cdot |\vcc{2}{G'}{U'}|$.
In order to show that our algorithm has $O^*(\ca^t)$ running time,
it suffices to show that the number of leaves in the branching tree of
this call is at most $\ca^{m_1(G',U')}$.
To show this we will use the following lemmas.
\begin{lemma}\label{lem:analysis}
For a recursive call $\wvcalg{G}{U}{w}$ in which $q = 2$, the
number of leaves in the branching tree of the call is
at most $(\sqrt{2})^{m(G,U)} \cdot \cb^{M(G,U)}$, where
$m(G,U) = |\vcc{2}{G}{U}|$ and
$M(G,U) = \sum_{v \in \vcc{2}{G}{U}} \map{G}{U}(v)$.
\end{lemma}
\begin{lemma}\label{lem:M}
If $U$ is a good vertex cover of $G$,
$M(G,U) \geq |\vcc{2}{G}{U}| - 3|\vcc{1}{G}{U}|$.
\end{lemma}
\begin{proof}
We prove the lemma by induction on $|\vcc{1}{G}{U}|$.
The base of the induction is true since $\map{G}{U}(v) = 1$ for all $v$
if $|\vcc{1}{G}{U}| = 0$, so $M(G,U) = |\vcc{2}{G}{U}|$.
If $|\vcc{1}{G}{U}| > 0$, pick $v \in \vcc{1}{G}{U}$ and
let $G' = G-v$ and $U' = U \setminus \{v\}$.
By the induction hypothesis,
$M(G',U') \geq |\vcc{2}{G'}{U'}| - 3|\vcc{1}{G'}{U'}|
= |\vcc{2}{G}{U}| - 3|\vcc{1}{G}{U}| + 3$.
We will show that $M(G,U) \geq M(G',U') - 3$ which will prove the lemma.
$v$ has at most 3 neighbors, and we will show that each neighbor decreases the
value of $M(G,U)$ by at most 1 compared to $M(G',U')$.

If $x$ is bad than either $x$ has a neighbor $u$ in $\vcc{2}{G}{U}$,
and let $u'$ be the unique neighbor of $u$ in $G[U]$.
If $u \in \vccs{2}{G}{U}$, $u$ is the only neighbor of $x$ in $\vcc{2}{G}{U}$.
Therefore, $x$ decreases the value of $\map{G}{U}(u)$ by at most 1 compared to
$\map{G'}{U'}(u)$ (namely, $\map{G}{U}(u) \geq \map{G'}{U'}(u)-1$), and
does not change the $\map{G}{U}$-values of the other vertices.
If $u \in \vcc{2}{G}{U} \setminus \vccs{2}{G}{U}$,
$x$ can be also adjacent to $u'$, but it does not have neighbors in
$\vcc{2}{G}{U} \setminus \{u,u'\}$.
Therefore, $x$ decreases the values of $\map{G}{U}(u)$ and $\map{G}{U}(u')$
by at most $\frac{1}{2}$,
and does not change the $\map{G}{U}$-values of the other vertices.
Therefore, in this case we also have that 
$x$ decreases the value of $M(G,U)$ by at most 1.

If $x$ is semi-bad then $x$ has two neighbors $u_1,u_2 \in \vcc{2}{G}{U}$.
If $u_1,u_2 \in \vccs{2}{G}{U}$ then $x$ decreases the values of
$\map{G}{U}(u_1)$ and $\map{G}{U}(u_2)$ by at most $\frac{1}{2}$,
and does not change the $\map{G}{U}$-values of the other vertices.
Therefore, $x$ decreases the value of $M(G,U)$ by at most 1.
It is also easy to verify that this is also true when one or two vertices
from $u_1,u_2$ are in $\vcc{2}{G}{U} \setminus \vccs{2}{G}{U}$.
\end{proof}

By Lemma~\ref{lem:analysis}, the number of leaves in the branching
tree of $\wvcalg{G'}{U'}{w'}$ is at most
$(\sqrt{2})^{m(G',U')} \cdot \cb^{M(G',U')}$.
By Lemma~\ref{lem:M} and since
$|\vcc{1}{G'}{U'}| < \beta \cdot |\vcc{2}{G'}{U'}|$, we have that
$M(G',U') \geq |\vcc{2}{G'}{U'}| - 3|\vcc{1}{G'}{U'}|
> (1-3\beta)\cdot m(G',U')$.
Therefore, the number of leaves in the branching tree of the call is
at most $(\sqrt{2}\cdot \cb^{1-3\beta})^{m(G',U')} \leq \ca^{m(G',U')} \leq
\ca^{m_1(G',U')}$.

We now prove Lemma~\ref{lem:analysis}. The proof uses induction
on the height of branching tree of the call.
Consider a call $\wvcalg{G}{U}{w}$.
If Rule~(\ref{rule:final}) is applied in this call then $M(G,U) = 0$.
The application of Rule~(\ref{rule:final}) decreases $m(G,U)$ by~2 in each branch.
Therefore, by the induction hypothesis, the number of leaves in the branching
tree of the call is at most
$2 \cdot (\sqrt{2})^{m(G,U)-2} = (\sqrt{2})^{m(G,U)}$.

Now consider a call in which Rule~(\ref{rule:vu1u2}) is applied.
Suppose that $u_2$ is defined,
and let $u'_2$ be the unique neighbor of $u_2$ in $G[U]$.
Rule~(\ref{rule:vu1u2}) generates four branches.
In the first three branches,
2 vertices from $\{u_1,u'_1,u_2,u'_2\}$ are deleted from the graph,
and the remaining 2 vertices are moved from $\vcc{2}{G}{U}$ to $\vcc{1}{G}{U}$.
Therefore, the value of $m(G,U)$ decrease by 4 in these branches.
Moreover, in the branch $G_2 = G-(N[u_1] \cup N[u_2])$,
$v$ has no neighbors in $V(G_2) \setminus U$, so
Rule~(\ref{rule:no-U-neighbors}) is applied on $v$, and then
Rule~(\ref{rule:degree-1}) is applied on $v$.
Thus, the value of $m(G,U)$ decreases by 6 ($v,u'_1,u'_2$ are deleted
from the graph and $v',u_1,u_2$ are moved from $\vcc{2}{G}{U}$ to
$\vcc{1}{G}{U}$).

We now bound the decrease in $M(G,U)$ in each of the four branches.
Consider the branch $G-\{u_1,u_2\}$.
Since $u'_1$ is moved from $\vcc{2}{G}{U}$ to $\vcc{1}{G}{U}$,
the value $\map{G}{U}(u'_1)$ will no longer be included in $M(G,U)$,
so this causes a decrease of at most 1 in $M(G,U)$.
Additionally, $u'_1$ can have two neighbors in $V(G)\setminus U$.
Since $u'_1$ is moved from $\vcc{2}{G}{U}$ to $\vcc{1}{G}{U}$,
each neighbor of $u'_1$ can cause a decrease of at most 1 in $M(G,U)$
(see the proof of Lemma~\ref{lem:M}).
Therefore, $u'_1$ causes a decrease of at most 3 in $M(G,U)$.
Similarly, $u'_2$ causes a decrease of at most 3 in $M(G,U)$.

The vertex $u_1$ is deleted from the graph, so the value $\map{G}{U}(u_1)$
will no longer be included in $M(G,U)$.
This decreases $M(G,U)$ by at most 1.
Since $u_1$ is deleted from the graph,
one or two neighbors of $u_1$ in $V(G)\setminus U$ can change from semi-bad
to bad vertices.
Suppose that exactly one neighbor $x$ changes from a semi-bad to a bad vertex.
This change can cause a decrease of at most $\frac{1}{2}$ in $M(G,U)$,
either by decreasing the $\map{G}{U}$-value of one neighbor of $x$ by
$\frac{1}{2}$, or by decreasing the $\map{G}{U}$-values of a neighbor $u$ of $x$
and of its unique neighbor $u'$ in $G[U]$ by $\frac{1}{4}$ each.
However, if $x$ change from a semi-bad to a bad vertex, then before the
application of the rule, $\map{G}{U}(u_1) \leq \frac{1}{2}$
(Note that if $u_1 \in \vcc{2}{G}{U}\setminus \vccs{2}{G}{U}$, then we actually
have $\map{G}{U}(u_1) \leq \frac{3}{4}$ and $\map{G}{U}(u'_1) \leq \frac{3}{4}$.
For the sake of the proof, we increase $\map{G}{U}(u'_1)$ by $\frac{1}{4}$
and decrease $\map{G}{U}(u_1)$ by $\frac{1}{4}$.
After this change, $\map{G}{U}(u_1) \leq \frac{1}{2}$).
Therefore, the decrease in $M(G,U)$ due to not including $\map{G}{U}(u_1)$
is at most $\frac{1}{2}$, and the total decrease in $M(G,U)$ due to the deletion
of $u_1$ is at most $\frac{1}{2}+\frac{1}{2} = 1$.
The decrease in $M(G,U)$ is also at most 1 in the case in which
two neighbors of $u_1$ change from semi-bad to bad vertices.
Similarly, $u_2$ causes a decrease of at most 1 in $M(G,U)$.
Therefore, the value of $M(G,U)$ decreases by at most 8 in this branch.

Next consider the branch $G-(N[u_1] \cup \{u_2\})$.
Since now $u'_1$ is deleted from the graph, the decreases in $M(G,U)$ due to
$u_1,u'_1,u_2,u'_2$ are at most $1,1,1,3$, respectively. Thus,
the value of $M(G,U)$ decreases by at most 6 in this branch.
Symmetrically, the value of $M(G,U)$ decreases by at most 6 in the branch
 $G-(\{u_1\} \cup N[u_2])$.
In the branch $G-(N[u_1] \cup N[u_2])$, each vertex from $u_1,u'_1,u_2,u'_2$
decreases the value of $M(G,U)$ by at most 1.
As described above, in this branch $v$ is deleted from the graph, and $v'$
is moved from $\vcc{2}{G}{U}$ to $\vcc{1}{G}{U}$.
These changes decrease $M(G,U)$ by at most 1 and 3, respectively.
The total decrease in $M(G,U)$ in this branch is at most 8.

By the induction hypothesis, the number of leaves in the branching
tree of the call $\wvcalg{G}{U}{w}$ is at most
$(\sqrt{2})^{m(G,U)-4} \cdot \cb^{M(G,U)-8} +
2(\sqrt{2})^{m(G,U)-4} \cdot \cb^{M(G,U)-6} +
 (\sqrt{2})^{m(G,U)-6} \cdot \cb^{M(G,U)-8}
\leq (\sqrt{2})^{m(G,U)} \cdot \cb^{M(G,U)}$.

In the case when $u_2$ is not defined, there are two branches:
$G-u_1$ and $G-N[u_1]$.
The decreases in $m(G,U)$ in these branches are 2 and 4, respectively,
and the decreases in $M(G,U)$ are at most 4 and 6, respectively.
By the induction hypothesis, the number of leaves in the branching
tree of the call $\wvcalg{G}{U}{w}$ is at most
$(\sqrt{2})^{m(G,U)-2} \cdot \cb^{M(G,U)-4} +
 (\sqrt{2})^{m(G,U)-4} \cdot \cb^{M(G,U)-6}
\leq (\sqrt{2})^{m(G,U)} \cdot \cb^{M(G,U)}$.

The analysis of Rule~(\ref{rule:vu1u2-2}) is similar to the analysis
of Rule~(\ref{rule:vu1u2}) (note that after the application of this rule,
Rule~(\ref{rule:triangle}) is applied in the branch $G-(N[u_1] \cup N[u_2])$)
and we omit the details.
We also need to consider Rules (2), (3), (\ref{rule:degree-1}),
and (\ref{rule:triangle}).
Note that we can ignore Rule~(\ref{rule:no-U-neighbors}) since at this stage
of the algorithm, after the application of this rule on a vertex $v$,
Rule~(\ref{rule:degree-1}) is applied on $v$.
The application of Rules (2), (\ref{rule:degree-1}), or (\ref{rule:triangle})
causes deletion of one or more vertices of the graph.
The effect of deleting a single vertex $v$ on the values $m(G,U)$ and $M(G,U)$
is as follows.
If $v \notin \vcc{2}{G}{U}$ then $m(G,U)$ does not change, and $M(G,U)$ does not
change or increases.
If $v \in \vcc{2}{G}{U}$ then $m(G,U)$ decreases by 2, and $M(G,U)$ decreases
by at most 4.
Therefore, if the rule deletes one vertex $v \in \vcc{2}{G}{U}$,
by the induction hypothesis the number of leaves in the branching
tree of the call $\wvcalg{G}{U}{w}$ is at most
$(\sqrt{2})^{m(G,U)-2} \cdot \cb^{M(G,U)-4} \leq
(\sqrt{2})^{m(G,U)} \cdot \cb^{M(G,U)}$.
More generally, if the rule deletes incrementally $l$ vertices from
$\vcc{2}{G}{U}$, we have that the number of leaves is at most
$(\sqrt{2})^{m(G,U)-2l} \cdot \cb^{M(G,U)-4l} \leq
(\sqrt{2})^{m(G,U)} \cdot \cb^{M(G,U)}$.
This completes the proof of Lemma~\ref{lem:analysis}.

\bibliographystyle{abbrv}
\bibliography{wvc}

\end{document}